\documentclass[12pt]{article}
\usepackage[left=3cm, right=3cm, top=3cm, bottom=3cm]{geometry}
\usepackage{amsmath,amsthm,amssymb,natbib,enumerate,xcolor,graphicx}
\usepackage[hyphens]{url}
\usepackage{hyperref}
\hypersetup{
    colorlinks,
    linkcolor={red!50!black},
    citecolor={blue!50!black},
    urlcolor={blue!50!black}
}
\newtheorem{thm}{Theorem}
\newtheorem{prop}[thm]{Proposition} 

\newtheorem{cor}[thm]{Corollary}
\theoremstyle{definition}
\newtheorem{defi}{Definition}
\theoremstyle{remark}

\newcommand*{\set}[1]{\left\{#1\right\}}

\title{Infection arbitrage}
\author{Sander Heinsalu\thanks{Research School of Economics, Australian National University, 
25a Kingsley St, Acton ACT 2601, Australia.
Email: sander.heinsalu@anu.edu.au, 
website: \url{https://sanderheinsalu.com/}
The author thanks MIT for its hospitality during the writing of this paper. 
}}
\date{\today}

\begin{document}
\maketitle
\begin{abstract}
Increasing the infection risk early in an epidemic is individually and socially optimal under some parameter values. The reason is that the early patients recover or die before the peak of the epidemic, which flattens the peak. This improves welfare if the peak exceeds the capacity of the healthcare system and the social loss rises rapidly enough in the number infected. The individual incentive to get infected early comes from the greater likelihood of receiving treatment than at the peak when the disease has overwhelmed healthcare capacity. 

Calibration to the Covid-19 pandemic data suggests that catching the infection at the start was individually optimal and for some loss functions would have reduced the aggregate loss. 

	
\end{abstract}


In a severe enough epidemic that is likely to overwhelm the capacity of the healthcare system, getting infected early is individually rational, in particular under most parameter values consistent with the data on the Covid-19 pandemic. 
The reason is that getting treated early when the healthcare system still has free capacity is better than facing rationing of medical services during the peak of the epidemic.
If the fraction infected at the peak is large enough, then the higher likelihood of treatment when catching the disease early outweighs the greater total probability of being infected when trying to catch it early (as opposed to reducing risk throughout the epidemic). 
The preference for early exacerbation holds for any health condition likely to need treatment at some point during the epidemic, not just for the disease causing the epidemic. 

Getting infected early imposes the obvious negative \emph{transmission externality} on others---the initial disease carriers spread the infection. However, the welfare effect of early infection is ambiguous because of a positive \emph{peak load externality}: the initially infected people recover or die before the peak of the epidemic, thus reduce the load on the healthcare system when its capacity is needed the most. 
The balance of the transmission and peak load externalities determines whether infecting part of the population initially improves welfare. 

If many people try to get infected early, then the peak of the epidemic arrives earlier and may be higher or lower than without attempted early infection. The more people who try to acquire the disease early, the lower the incentive to acquire it initially, because the less spare healthcare capacity remains early on. Symmetrically, this arbitrage argument shows that the more people who try to delay the infection before the peak, the lower the incentive to protect oneself during that time. The strategic substitutability of risk choices leads to a mixed equilibrium in which some fraction of the population arbitrages the infection timing by getting infected early. The rest best respond by doing nothing until the fraction infected is large enough, at which point everyone tries to reduce their risk until the peak of the epidemic passes. 

The arbitrage of infection timing flattens the peak of the epidemic and spreads it over a longer time. The welfare effect of this depends on the loss function. If the loss rises fast in the fraction infected, then welfare depends mostly on the height of the peak. In this case, reducing the maximal prevalence of the disease reduces aggregate losses. If the loss increases slowly in the fraction infected, then the wider and lower peak under early infection is worse than the higher and later peak under the standard policy of reducing risk. 

The literature on theoretical epidemiology is vast, starting from \cite{kermack+mckendrick1927}. The study of the incentives to protect oneself goes back at least to \cite{peltzman1975}, although in a safety context. An early review of economic epidemiology is the book of \cite{philipson+posner1993}. 

\cite{toda2020} finds that delaying infection control efforts until a significant fraction of the population is infected is socially optimal. The present work takes the natural next step and shows that the opposite intervention to the usual risk reduction is individually and socially optimal in the early stage of a severe epidemic. 

The individual and social optimality of protection from disease are also studied in \cite{auld1996,auld2006,chen+toxvaerd2014,fenichel2013,galeotti+rogers2013,goyal+vigier2015,heinsalu2020,kremer1996,kremer+morcom1998,
rowthorn+toxvaerd2017,talamas+vohra2019,toxvaerd2019}. To the author's knowledge, the literature has not considered deliberate early infection as a feasible (let alone optimal) individual choice or policy instrument. 

Mathematically, the environment with a changing hazard rate of an event (infection) over time resembles \cite{khan+stinchcombe2015}. The current paper studies a game in addition to an individual decision. Each agent solves a control problem (chooses risk every period) instead of a stopping problem (when to vaccinate). 

The next section examines individual incentives to reduce or increase risk, and calibrates the results to the Covid-19 data from the empirical literature. Section~\ref{sec:interacting} studies infection arbitrage by interacting decision makers. Welfare comparisons using the compartmental susceptible-infected-recovered model are in Section~\ref{sec:SIR}. Section~\ref{sec:contt} extends the individual and equilibrium results to continuous time. A discussion of policy implications is in Section~\ref{sec:conclusion}.

\section{A two-period example}
\label{sec:twoperiod}

A continuum of individuals each face a risk of infection $r_1\in(0,1)$ in period $t=1$ and if this risk does not realize, then a greater risk $r_2\in(r_1,1)$ in period $2$. The loss $\ell_t>0$ conditional on infection also increases over time: $\ell_2 >\ell_1$. The loss $\ell_t$ summarizes the expected discounted value of all negative consequences of the risk. Rising risk and loss over time reflect an epidemic spreading in the population, which both increases the infection probability and leaves fewer medical resources per patient, thus worsens the expected consequences of infection. The improvement in treatment or the availability of supplies is assumed slower than the spread of the epidemic, and is incorporated in $\ell_t$ w.l.o.g.  

Before being infected, the individual can take an action $a_{t}\in \set{\check{a}_t,1,\hat{a}_t}$, with $0<\check{a}_t<1<\hat{a}_t<\frac{1}{r_t}$, to decrease or increase the infection risk to $a_tr_t$ that period. Action $a_t=1$ is costless and leaves the risk unchanged, $\check{a}_t$ costs $\check{c}\in(0,r_1\ell_1)$ and decreases the risk, and $\hat{a}_t$ costs $\hat{c} >0$ and increases the risk. 
Example actions reducing infection risk are social distancing measures, using protective equipment, following hygiene practices. 
Example actions increasing the risk are doing (more) work in healthcare, contaminated waste disposal, or in crowded locations, or deliberately obtaining the disease agent and infecting oneself. 

In the second period (at the peak of the epidemic), the individual does not increase the risk, because $-\hat{a}_{2}r_2\ell_2 -\hat{c} <-r_2\ell_2 $. Decreasing the risk is optimal if $-\check{a}_{2}r_2\ell_2 -\check{c} >-r_2\ell_2 $, which reduces to 
\begin{align}
\label{decr2}
\check{c} <(1-\check{a}_{2})r_2\ell_2.
\end{align} 
If $\check{c} \geq(1-\check{a}_{2})r_2\ell_2$, then the individual leaves the risk unchanged. 
Only the continuation value $\pi_2^* $ for an uninfected person at the start of period $2$ matters for subsequent results. 
The continuation payoff of an infected individual is zero, because the loss from infection has already been incorporated into the first period payoff. 

Suppose $\check{c} <(1-\check{a}_{2})r_2\ell_2$, interpreted as a large risk and loss from catching the disease at the height of the epidemic, high prevention effectiveness $1-\check{a}_{2}$, or a small cost of preventive measures. 
Then $\pi_2^* =-\check{a}_{2}r_2\ell_2 -\check{c}$. Otherwise, the individual optimally does nothing and obtains $\pi_2^* =-r_2\ell_2$. 

In the first period, the individual compares doing nothing (payoff $-r_1\ell_1+(1-r_1)\pi_2^*$), reducing the risk (payoff $-\check{a}_{1}r_1\ell_1 -\check{c} +(1-\check{a}_{1}r_1)\pi_2^*$) and increasing the risk ($-\hat{a}_{1}r_1\ell_1 -\hat{c} +(1-\hat{a}_{1}r_1)\pi_2^*$). 
Raising one's risk is optimal if 
$(1-\hat{a}_{1})r_1\ell_1 +(1-\hat{a}_{1})r_1\pi_2^* >\hat{c}$ and $(\check{a}_{1}-\hat{a}_{1})r_1\ell_1 +(\check{a}_{1}-\hat{a}_{1})r_1\pi_2^* > \hat{c}-\check{c}$. 
These are equivalent to  
$-r_1(\ell_1 +\pi_2^*) >\frac{\hat{c}}{\hat{a}_{1}-1}$ because $\hat{c}-\check{c} <\hat{c}$ and $\hat{a}_{1}-\check{a}_{1} >\hat{a}_{1}-1$. 
Recall that $\pi_2^*<0$. 
If the cost of increasing one's risk is small, the risk is large and can be significantly increased and the second-period losses from infection and mitigation $|\pi_2^*|$ are large enough, then the individual prefers to increase risk initially, because this reduces the risk later when it is more costly. 
The logic is similar to vaccination with a live attenuated vaccine which causes disease with a small positive probability (e.g., most poliomyelitis cases were caused by the vaccine over several years until the vaccine was phased out and the wild-type disease made a comeback \citep{macklin+2020}).\footnote{
Another analogy is counter-firing against forest fires (\url{https://www.nwcg.gov/term/glossary/counter-fire}) and controlled burns to prevent large fires. 
}
Vaccination is still optimal, because it reduces the infection risk over several years, and this risk when unvaccinated is substantially larger than the likelihood of catching the disease from the vaccine. 


In terms of the primitive parameters, if the individual plans to reduce risk in period $2$, then in period $1$, increasing the risk is optimal if 
\begin{align}
\label{incrdecr}
r_1(-\ell_1 +\check{a}_{2}r_2\ell_2 +\check{c}) >\frac{\hat{c}}{\hat{a}_{1}-1}.
\end{align} 
If the individual anticipates doing nothing in period $2$, then increasing the risk in period $1$ is optimal when 
\begin{align}
\label{incrconst}
r_1(-\ell_1 +r_2\ell_2) >\frac{\hat{c}}{\hat{a}_{1}-1}.
\end{align} 
The optimality of increasing one's risk depends on the parameter values. The next subsection uses empirical estimates of the epidemic characteristics of the novel coronavirus to calibrate~(\ref{incrdecr}) and~(\ref{incrconst}). For the median estimates, increasing one's risk early in the epidemic is individually rational. The utility and informational assumptions underlying the individual choices are discussed after the calibration. 
Section~\ref{sec:dynamiccalibration} calibrates the individual best response in a continuous time epidemiological model and shows that infecting oneself early remains the optimal strategy.

\subsection{Calibration of the individual decision}
\label{sec:calibration}

Interpret period $2$ as the peak of the epidemic, with the least resources available per patient. Period $1$ may be any time before the peak when the healthcare system still has spare capacity for the epidemic disease carriers. 

The infection rate in the population at the peak of the epidemic is estimated to be 40--70\% \citep{shlain2020}, similarly to the rate \cite{meltzer+2015} estimate for an influenza pandemic. 
Assume no deaths among the 70\% of cases who are non-severe and non-critical. Assume 4.67\% of cases are critical, as in the data of \cite{livingston+bucher2020}. Among critical cases, 49\% survived in the data of \cite{wu+mcgoogan2020}, 38.5\% in \cite{yang+2020} and 21\% in \cite{zhou+2020}. Assume all critical cases die without a ventilator. This is a lower bound on the case-fatality rate when the healthcare system is overwhelmed, because severe cases become critical under a shortage of medical professionals and supplies. An upper bound on the case-fatality rate is that all severe and critical cases die.
 
An estimated 160\,000 ventilators are available in the US as of 1 April 2020 \citep{JHU2020}, which would cover 2.45\% of the critical cases (who form 2\% of the US population of 329 million\footnote{\url{https://www.census.gov/popclock/}} if 5\% of infected are critical and 40\% of the population is infected) at the peak of the epidemic. Getting infected at the peak of the epidemic then results in an expected death probability between $0.0467(1 -0.0245\cdot 0.51)=0.046$ and $0.3(1 -0.0245\cdot 0.21)=0.298$, which is the product of the probability of needing a ventilator and the probability of either not obtaining one or dying even with ventilation. The large upper bound $0.298$ assumes severe cases become critical without medical intervention that is unavailable during the peak of the epidemic, and die after becoming critical. 
Catching the disease early when medical capacity is unconstrained results in an expected death probability $0.0467(1 -0.51) =0.023$ \citep{onder+2020}. 
If expected utility is linear in the death probability, then the loss ratio $\ell_2/\ell_1$ between being infected at the peak of the epidemic and initially is in the range $[\frac{0.046}{0.023},\frac{0.298}{0.023}]$, i.e.,\ from $2$ to $13$. 

Trying to catch the disease deliberately likely has close to 100\% success rate, so in~(\ref{incrconst}), $\hat{a}_1r_1\approx1$. The fraction infected is close to zero early in an epidemic, which makes~(\ref{incrconst}) approximately $-\ell_1 +r_2\ell_2 >\hat{c}$ and~(\ref{incrdecr}) to $-\ell_1 +\check{a}_{2}r_2\ell_2 +\check{c} >\hat{c}$. Compared to the loss of dying with probability over $1.38\%$ \citep{verity+2020}, the cost of increasing or decreasing one's risk is likely small, in which case~(\ref{incrconst}) is approximately $r_2\ell_2>\ell_1$ and~(\ref{incrdecr}) is $\check{a}_{2}r_2\ell_2 >\ell_1$. \cite{shlain2020,meltzer+2015} predict a peak infection rate of 40--70\% even with control measures, so $\check{a}_{2}r_2\geq 0.4$. 
Trying to catch the disease early in an epidemic is individually rational if $0.4 \ell_2/\ell_1\geq 1$, i.e., $\ell_2\geq 2.5\ell_1$, which is at the lower end of the range $[2,13]$ of loss ratios. 

If expected utility decreases faster than linearly in the death probability, then $\ell_2/\ell_1$ increases and deliberate infection early becomes relatively more attractive. One cause of the faster decrease in expected utility is risk aversion, because all death probabilities are below 50\%, so the larger probability is riskier (creates a Bernoulli random variable with greater variance). 
Prospect theory predicts upward distortion of small probabilities, in which case an increase in the likelihood of death translates into a slower than linear reduction in the payoff. Then infecting oneself becomes relatively less attractive. 

Information is imprecise early in an epidemic, so individuals have an option value of waiting, which reduces their motive to increase their risk. 

In a less severe epidemic where the peak infection rate in the US is below $\frac{1}{13}$ of the population (the China scenario), increasing one's risk is not optimal. A lower fraction of critical and severe cases has a similar effect as reducing the peak infection rate. 

If the cost of infecting oneself is large, perhaps due to effective isolation of disease carriers early in an epidemic, then the cost may outweigh the benefit of increasing one's risk. Conversely, if the cost of reducing one's risk is large, then individuals optimally do nothing in period $2$, so their probability of catching the disease at the peak of the epidemic is $r_2$, higher than $\check{a}_2r_2$ for any population infection rate. This increases the motive to increase risk in period $1$. 

A necessary condition for increasing one's risk to be optimal is that the loss increases over time. Otherwise, a non-discounting agent is indifferent to the timing of risk and a discounting agent strictly prefers to delay the infection. 
After the peak of the epidemic, the risk and loss fall over time. For general decision problems with a decreasing hazard rate, the results of \cite{khan+stinchcombe2015} imply that it is always optimal to reduce one's risk and delay the disease. 

If many people decide to get infected early when medical help is still available, then of course the peak of infection arrives early. This gives individuals the incentive to delay their disease past the peak, which in turn delays the peak. Conditional on catching the illness, rational decision makers want it when others are healthy. 
The next section discusses how arbitrage tends to equalize the net loss (the cost of changing the risk, plus the product of the risk and the loss) across periods. 


\section{Infection arbitrage by interacting decision makers}
\label{sec:interacting}

To model the interaction of risk choices, let $r_t$ and $l_t$ in~(\ref{incrdecr}) and~(\ref{incrconst}) increase in the fraction $\hat{\sigma}_t$ of individuals choosing $\hat{a}_t$ and decrease in the fraction $\check{\sigma}_t$ taking $\check{a}_t$. The fall in period $2$ risk and loss when more others choose to decrease risk, as~(\ref{decr2}) shows, weakens the incentive for a given individual to reduce risk in period $2$ and to increase it in period $1$. Similarly, the increase in period $1$ risk and loss when others choose to raise their risk decreases an individual's incentive to increase risk in period $1$ according to~(\ref{incrdecr}) and~(\ref{incrconst}). 

The risk and loss in the second period may also decrease in the fraction infected in period $1$, but any effect across periods is assumed weaker than the corresponding effect within any period. 
Formally, the risk and loss are functions of the action history: $r_t =r_{tf}(\sigma^t)$ and $\ell_t =\ell_{tf}(\sigma^t)$, where $\sigma^t =(\hat{\sigma}_{\tau},\check{\sigma}_{\tau})_{\tau=1,\ldots,t}$. Assume that for $f_t\in\set{r_{tf},\ell_{tf}}$, 
$\frac{\partial f_t}{\partial \hat{\sigma}_t} > \frac{\partial f_2}{\partial \hat{\sigma}_1}\geq 0\geq \frac{\partial f_2}{\partial \check{\sigma}_1} > \frac{\partial f_t}{\partial \check{\sigma}_t}$. 
The assumption $\hat{a}_t<\frac{1}{r_t}$ is modified to $\hat{a}_t<\frac{1}{\max r_{tf}(\sigma^t)}$. 
The interpretation is that infection confers (partial) immunity and a period is longer than the course of illness. This further strengthens the anticoordination motive to catch the disease when others are healthy.

\begin{defi}
\label{def:equil}
A subgame perfect equilibrium consists of choices $\sigma^* =(\hat{\sigma}_t^*,\check{\sigma}_t^*)_{t=1,2}^{}$, risks $r^*=(r_t^*)_{t=1,2}^{}$ and losses $\ell^*=(\ell_t^*)_{t=1,2}^{}$ each period satisfying the following. 
\begin{enumerate}[(a)]
\item 
$\hat{\sigma}_2^*=0$. 
\item 
If $\check{c} <(1-\check{a}_{2})r_2^*\ell_2^*$, then $\check{\sigma}_2^*=1$, but if $\check{c} >(1-\check{a}_{2})r_2^*\ell_2^*$, then $\check{\sigma}_2^*=0$. 
\item 
\label{item:hatsigma}
If $r_1^*(-\ell_1^* +\check{a}_{2}r_2^*\ell_2^* +\check{c}) >\frac{\hat{c}}{\hat{a}_{1}-1}$ and $\check{c} \leq(1-\check{a}_{2})r_2^*\ell_2^*$ or $r_1^*(-\ell_1^* +r_2^*\ell_2^*) >\frac{\hat{c}}{\hat{a}_{1}-1}$ and $\check{c} \geq(1-\check{a}_{2})r_2^*\ell_2^*$, then $\hat{\sigma}_1^* =1$. However, if $r_1^*(-\ell_1^* +\check{a}_{2}r_2^*\ell_2^* +\check{c}) <\frac{\hat{c}}{\hat{a}_{1}-1}$ and $\check{c} \leq(1-\check{a}_{2})r_2^*\ell_2^*$ or $r_1^*(-\ell_1^* +r_2^*\ell_2^*) <\frac{\hat{c}}{\hat{a}_{1}-1}$ and $\check{c} \geq(1-\check{a}_{2})r_2^*\ell_2^*$, then $\hat{\sigma}_1^* =0$. 
\item 
\label{item:checksigma}
If $r_1^*(\ell_1^* -\check{a}_{2}r_2^*\ell_2^* -\check{c}) >\frac{\check{c}}{1-\check{a}_{1}}$ and $\check{c} \leq(1-\check{a}_{2})r_2^*\ell_2^*$ or $r_1^*(\ell_1^* -r_2^*\ell_2^*) >\frac{\check{c}}{1-\check{a}_{1}}$ and $\check{c} \geq(1-\check{a}_{2})r_2^*\ell_2^*$, then $\check{\sigma}_1^* =1$. If the strict inequalities reverse, then $\check{\sigma}_1^* =0$.
\item
$r_t^*=r_	{tf}(\sigma^{t*})$ and $\ell_t^*=\ell_{tf}(\sigma^{t*})$. 
\end{enumerate}
\end{defi}
Subgame perfect equilibria are just called equilibria henceforth. 

Assume for simplicity that the period $2$ risk and loss cannot be reduced enough to make doing nothing optimal even if all agents increase risk in period $1$ and decrease it in period $2$. Formally, $\check{c} \leq(1-\check{a}_{2})\min_{\sigma}\{r_{2f}(\sigma)\ell_{2f}(\sigma)\}$, where $\arg\min_{\sigma}\{r_{2f}(\sigma)\ell_{2f}(\sigma)\}$ is $\hat{\sigma}_1=1$, $\check{\sigma}_1=0$, $\hat{\sigma}_2=0$, $\check{\sigma}_2=1$. 
Then all equilibria feature $\check{\sigma}_2^*=1$ and thus the continuation payoff may be written as 
$\pi_2^*(\sigma_1) =-\check{a}_{2}r_{2f}(\hat{\sigma}_1,\check{\sigma}_1,0,1)\ell_{2f}(\hat{\sigma}_1,\check{\sigma}_1,0,1) -\check{c}$. 


The following proposition characterizes the equilibrium set. 
\begin{prop}
\label{prop:equilibria}
\begin{enumerate}
\item
Doing nothing ($a_1=1$) is an equilibrium if 
\\$-\frac{\hat{c}}{\hat{a}_{1}-1} \leq r_{1f}(0,0)(\ell_{1f}(0,0) +\pi_2^*(0,0)) \leq\frac{\check{c}}{1-\check{a}_{1}}$. 
\item
\label{item:mixdecrease}
Mixing between decreasing the risk and doing nothing is an equilibrium if 
\\$r_{1f}(0,1)(\ell_{1f}(0,1) +\pi_2^*(0,1)) <\frac{\check{c}}{1-\check{a}_{1}} <r_{1f}(0,0)(\ell_{1f}(0,0) +\pi_2^*(0,0))$. 
\item 
\label{item:mixincrease}
Mixing between increasing the risk and doing nothing is an equilibrium if 
\\$-r_{1f}(1,0)(\ell_{1f}(1,0) +\pi_2^*(1,0)) < \frac{\hat{c}}{\hat{a}_{1}-1} <-r_{1f}(0,0)(\ell_{1f}(0,0) +\pi_2^*(0,0))$. 
\item
\label{item:decreasing}
Decreasing the risk is an equilibrium if $r_{1f}(0,1)(\ell_{1f}(0,1) +\pi_2^*(0,1)) \geq\frac{\check{c}}{1-\check{a}_{1}}$. 
\item 
\label{item:increasing}
Increasing the risk is an equilibrium if $-r_{1f}(1,0)(\ell_{1f}(1,0) +\pi_2^*(1,0)) \geq\frac{\hat{c}}{\hat{a}_{1}-1}$. 
\end{enumerate} 
Equilibrium is generically unique. 
\end{prop}
The proof is merely checking the conditions in Definition~\ref{def:equil} and is thus omitted. 

It is reasonable to assume that if everyone tries to get infected in the first period, then the risk and loss are so large that individuals prefer not to increase their risk---sufficient for this is $\ell_{1f}(1,0)\geq \pi_2^*(1,0)$. Similarly, if all others choose to reduce risk, then an individual prefers either doing nothing or increasing own risk, sufficient for which is $\ell_{1f}(0,1)\leq \pi_2^*(0,1)$. The reason is that imperfect protective measures by the rest of society do not eliminate the disease, but reduce transmission enough that paying the cost $\check{c}$ of risk reduction is individually suboptimal. 
Under these assumptions, the candidate equilibria~\ref{item:decreasing}\ and~\ref{item:increasing}\ in Proposition~\ref{prop:equilibria} are not equilibria. Then the conditions guaranteeing mixed equilibria reduce to $\frac{\check{c}}{1-\check{a}_{1}} <r_{1f}(0,0)(\ell_{1f}(0,0) +\pi_2^*(0,0))$ for decreasing the risk (part~\ref{item:mixdecrease}) and 
$\frac{\hat{c}}{\hat{a}_{1}-1} <-r_{1f}(0,0)(\ell_{1f}(0,0) +\pi_2^*(0,0))$ for increasing the risk (\ref{item:mixincrease}). 

Focus on the parameter values at which the (negative) continuation payoff $\pi_2^*(0,0)$ from doing nothing in period $1$ is larger in absolute value than the loss $\ell_{1f}(0,0)$ in period $1$. The interpretation is that the expected consequences of catching the disease at the peak of the epidemic are worse than when catching it now with certainty. For example, medical supplies run out, many healthcare professionals fall ill and the rest are overloaded with patients, reducing the probability that a given person receives adequate treatment. If treatment is effective enough and the probability that the epidemic exhausts medical capabilities large enough, then getting infected early is optimal. 

When $\ell_{1f}(0,0) <|\pi_2^*(0,0)|$, no equilibrium involves decreasing the risk in the first period. For a small enough cost $\hat{c}$, the unique equilibrium is that a fraction of the population increases their risk and the rest do nothing. 
This is formalized in the following corollary. 
\begin{cor}
\label{cor:increase}
If $-r_{1f}(1,0)(\ell_{1f}(1,0) +\pi_2^*(1,0)) < \frac{\hat{c}}{\hat{a}_{1}-1} <-r_{1f}(0,0)(\ell_{1f}(0,0) +\pi_2^*(0,0))$, then in the unique equilibrium, $\check{\sigma}_1=0$ and $\hat{\sigma}_1^*$ solves $-r_{1f}(\hat{\sigma}_1^*,0)(\ell_{1f}(\hat{\sigma}_1^*,0) +\pi_2^*(\hat{\sigma}_1^*,0)) = \frac{\hat{c}}{\hat{a}_{1}-1}$. 
\end{cor}

The arbitrage of infection timing decreases the difference between the losses in different periods relative to doing nothing or reducing risk initially: 
$|\pi_2^*(\hat{\sigma}_1,\check{\sigma}_1)|-\ell_{1f}(\hat{\sigma}_1,\check{\sigma}_1) >|\pi_2^*(\hat{\sigma}_1^*,0)| -\ell_{1f}(\hat{\sigma}_1^*,0) >0$ for any $\hat{\sigma}_1<\hat{\sigma}_1^*$ and any $\check{\sigma}_1$. 
Corollary~\ref{cor:increase} additionally implies $|\pi_2^*(\hat{\sigma}_1^*,0)| -\ell_{1f}(\hat{\sigma}_1^*,0) =\frac{\hat{c}}{r_{1f}(\hat{\sigma}_1^*,0)\hat{a}_{1}-r_{1f}(\hat{\sigma}_1^*,0)}$, which goes to zero in $\hat{c}$, with the implication that if the costs of arbitrage are small, then so is the difference between the initial loss and the expected future loss. 
Arbitrage thus `flattens out' the peak of the epidemic in terms of both infection risk and the loss conditional on infection. For a single agent, this improves payoff by revealed preference. However, strategic interaction may yield a positive or negative welfare effect. The following proposition provides sufficient conditions for increases and decreases in welfare from flattening the peak of the epidemic by early infection, compared to doing nothing. 
\begin{prop}
\label{prop:welfare}
There exists $\Delta_{r\ell2}>0$ s.t.\ if 
\begin{align*}
r_{2f}(0,0,0,1)\ell_{2f}(0,0,0,1) -r_{2f}(1,0,0,1)\ell_{2f}(1,0,0,1)\leq \Delta_{r\ell2},
\end{align*} 
then welfare in the equilibrium in which individuals in period $1$ mix to increase their risk is less than when all individuals do nothing. 
There exists $\Delta_{r\ell1}>0$ s.t.\ if 
\begin{align*}
r_{1f}(1,0)\ell_{1f}(1,0) -r_{1f}(0,0)\ell_{1f}(0,0) <\Delta_{r\ell1}\text{ and }\pi_2(1,0) >\pi_2(0,0),
\end{align*} 
then welfare in the equilibrium in which individuals in period $1$ mix to increase their risk is greater than when all individuals do nothing. 
\end{prop}
\begin{proof}
$\pi_2(1,0) =-\check{a}_{2}r_{2f}(1,0,0,1)\ell_{2f}(1,0,0,1) -\check{c}$, so if $r_{2f}(0,0,0,1)\ell_{2f}(0,0,0,1) -r_{2f}(1,0,0,1)\ell_{2f}(1,0,0,1)\leq \Delta_{r\ell2}$, then $\pi_2(1,0) -\pi_2(0,0)\leq \check{a}_{2}\Delta_{r\ell2}$. 
The indifference condition 
\begin{align*}
&-\hat{a}_1r_{1f}(\hat{\sigma}_1^*,0)\ell_{1f}(\hat{\sigma}_1^*,0) -\hat{c} +(1-\hat{a}_1r_{1f}(\hat{\sigma}_1^*,0))\pi_2^*(\hat{\sigma}_1^*,0) 
\\& =-r_{1f}(\hat{\sigma}_1^*,0)\ell_{1f}(\hat{\sigma}_1^*,0) +(1-r_{1f}(\hat{\sigma}_1^*,0))\pi_2^*(\hat{\sigma}_1^*,0)
\end{align*} 
of the mixed equilibrium 
implies that the average payoff in it is $-r_{1f}(\hat{\sigma}_1^*,0)\ell_{1f}(\hat{\sigma}_1^*,0) +(1-r_{1f}(\hat{\sigma}_1^*,0))\pi_2^*(\hat{\sigma}_1^*,0)$. 
The average payoff if everyone does nothing in the first period is 
$-r_{1f}(0,0)\ell_{1f}(0,0) +(1-r_{1f}(0,0))\pi_2^*(0,0)$. The 
difference between the average equilibrium payoff and the average payoff when everyone does nothing is 
\begin{align}
\label{payoffdiff}
&r_{1f}(0,0)\ell_{1f}(0,0)-r_{1f}(\hat{\sigma}_1^*,0)\ell_{1f}(\hat{\sigma}_1^*,0) +(1-r_{1f}(\hat{\sigma}_1^*,0))\pi_2^*(\hat{\sigma}_1^*,0) -(1-r_{1f}(0,0))\pi_2^*(0,0),
\end{align} 
which is continuous. 

The risk and loss increase in period $1$, so if $|\pi_2(1,0)-\pi_2(0,0)|\rightarrow 0$ (the continuation payoff becomes constant in $\sigma_1$), 
then~(\ref{payoffdiff}) converges to  
\begin{align*}
&r_{1f}(0,0)\ell_{1f}(0,0)-r_{1f}(\hat{\sigma}_1^*,0)\ell_{1f}(\hat{\sigma}_1^*,0) +[r_{1f}(0,0)-r_{1f}(\hat{\sigma}_1^*,0)]\pi_2^*(\hat{\sigma}_1^*,0)
\\&<
[r_{1f}(0,0)-r_{1f}(\hat{\sigma}_1^*,0)][\ell_{1f}(\hat{\sigma}_1^*,0) +\pi_2^*(\hat{\sigma}_1^*,0)]
=[r_{1f}(0,0)-r_{1f}(\hat{\sigma}_1^*,0)]\frac{\hat{c}}{\hat{a}_{1}-1} <0, 
\end{align*} 
where the equality comes from Corollary~\ref{cor:increase}. 

If $|r_{1f}(0,0)\ell_{1f}(0,0)- r_{1f}(1,0)\ell_{1f}(1,0)|\rightarrow 0$, then by implication, $|r_{1f}(0,0)\ell_{1f}(0,0)- r_{1f}(\hat{\sigma}_1^*,0)\ell_{1f}(\hat{\sigma}_1^*,0)|\rightarrow 0$, so the payoff difference~(\ref{payoffdiff}) converges to 
$(1-r_{1f}(0,0))[\pi_2^*(\hat{\sigma}_1^*,0) -\pi_2^*(0,0)] >0$. 
\end{proof}

Unsurprisingly, increasing the risk and loss initially without much compensating decrease later worsens welfare. 
On the other hand, the equilibrium (increasing risk initially) is better than doing nothing when the risk and loss do not rise much at the start of the epidemic but fall during the peak. 
The welfare-maximizing strategy depends on the generally nonlinear $r_{tf}$ and $\ell_{tf}$. For any $\sigma$ in which $\hat{\sigma}_t\check{\sigma}_t=0$ (risk is not simultaneously increased and decreased), there may exist risk and loss functions making this $\sigma$ socially optimal. 
For generic risk and loss functions, no equilibrium is efficient, i.e., the welfare-maximizing strategy is not an equilibrium. The reason is the externalities of transmission and peak load imposed by individual decisions. 

To derive clearer welfare comparisons and policy implications, more structure needs to be imposed on the risk and loss functions. The next section derives the risk from the standard SIR model. The loss is derived from dividing the medical capacity by the number infected.

\subsection{Welfare in the SIR model with early infection}
\label{sec:SIR}

The risk function is derived from the standard continuous time SIR compartmental model of \cite{kermack+mckendrick1927}, reformulated in terms of fractions of the population instead of absolute numbers of people. Fraction $S_t$ is susceptible, $I_t$ infected and $R_t=1-S_t-I_t$ recovered, resistant or dead at time $t$. These fractions change over time as $\frac{dS_t}{dt} =-\beta S_tI_t$, $\frac{dI_t}{dt} =\beta S_tI_t-\gamma I_t$ and $\frac{dR_t}{dt} =\gamma I_t$. 
\cite{harko+2014} solve for the paths of change of the susceptible, infected and resistant fractions given initial conditions $(S_0,I_0,R_0)>0$. 


For technical convenience, the SIR model assumes the recovery times of infected people are exponentially distributed. This is unrealistic, because the duration of infection usually has an upper bound.\footnote{\cite{zhou+2020} find the longest duration of viral shedding for Covid-19 is 37 days. This is for a hospitalized case; the duration with mild symptoms is likely shorter.} If the recovery times are exponentially distributed, then increasing risk is never optimal, because the peak of the epidemic from a higher initial infected fraction is always higher and occurs earlier \citep{toda2020}. 
However, with bounded recovery times, infecting part of the population early reduces the height of the peak of the epidemic when the initially infected recover significantly before the peak (which now occurs earlier). There is no closed form solution to the SIR model with bounded recovery times. Figure~\ref{fig:init} shows a numerical example\footnote{
Mathematica code for the example is available on the author's website \url{https://sanderheinsalu.com/}
} 
in which deliberate early infection yields higher welfare. The parameters of the epidemic can be interpreted as incorporating the effect of standard infection control measures, in which case the comparison is between deliberate infection initially, followed by the control policy that would be optimal if early infection was not feasible, and this optimal control without early infection. 

%

\begin{figure}
\caption{Susceptible (blue curve), infected (orange) and recovered (green) fractions of the population without deliberate infection (dashed curves) and with initially infecting 5\% (solid curves). Parameters: $\beta=0.4$, $\gamma=0.1$, average duration of disease $10$ days.}
\label{fig:init}
\includegraphics[width=0.8\textwidth]{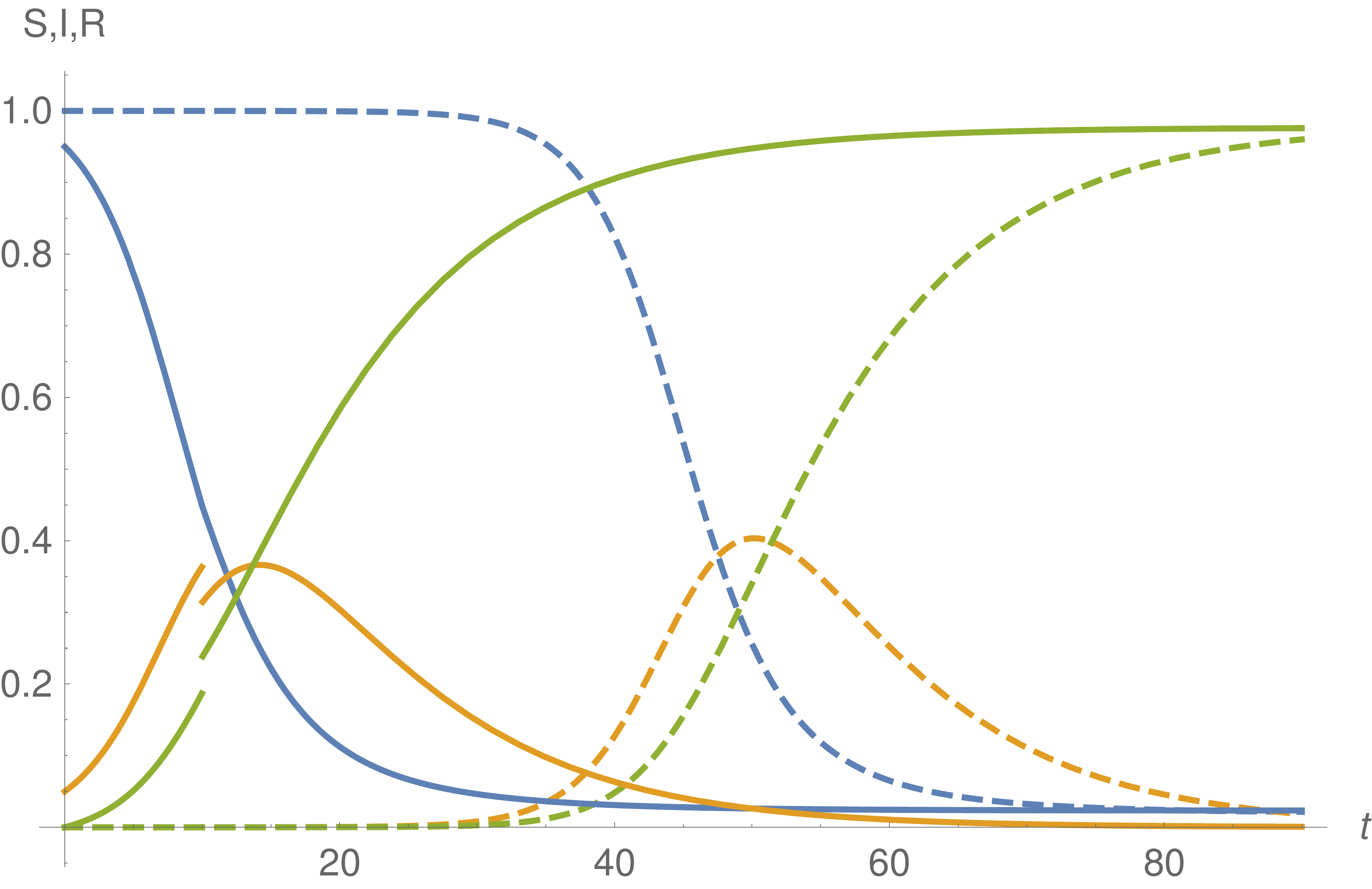}
\end{figure}

For simplicity and easier visualization, the duration of infection is assumed exactly 10 days, so on day 11 all the initially infected people recover. A bounded distribution of infection times would smooth out the downward jump in the fraction infected without changing the qualitative insight. In Figure~\ref{fig:init}, with early deliberate infection the epidemic peaks on days 10 and 14 on which 37\% of the population is infected. Without early infection, the single peak on day $50$ has 40\% of the population infected. 

Of course, if the fraction infected early is large enough, then both peaks of the epidemic are higher than without early infection. The optimal policy needs to be precisely tailored to the parameters of the disease, which is difficult especially for newly emerging infections about which information is limited. 

The SIR model assumes random mixing, so the early infected infect others. If they could be quarantined until they recover, then this increases the welfare gain from deliberate early infection relative to standard policies. Because early infection is also individually optimal, volunteers may be found who are sufficiently altruistic to agree to quarantine in exchange for getting access to the disease agent early. 

If welfare (the negative of the aggregate loss) only depended on the height of the peak, then the policy comparison in Figure~\ref{fig:init} would be clear. However, the loss also depends on how long the infected fraction stays above the capacity of the healthcare system, and how high above. To account for these effects, the loss function is next defined explicitly. Welfare is the negative of the integral of the loss function over the horizon of the epidemic. 

The per-person loss function is assumed constant (normalized to $1$) if the infected fraction $I$ of the population is less than the capacity $K$ of the medical system. If $I>K$, then the loss is $1+\lambda (1-\frac{K}{I})$, with $\lambda>0$ the increase in per-person loss if the capacity were zero. The interpretation is that $\lambda$ is the extra loss from not getting treatment conditional on infection and $1-\frac{K}{I}$ is the probability of not getting treatment due to rationing of the capacity. 


Welfare is   
$-\int_0^{t_{\max}}I_t[1+\text{\textbf{1}}\set{I_t >K}\lambda (1-\frac{K}{I_t})]dt$ over the horizon $t_{\max}$ of the epidemic. As long as $t_{\max}$ is substantially later than the peak of the epidemic, its value (here taken to be $90$ days) does not materially influence the results, because after the peak, the fraction infected declines at an exponential rate. 

As in Section~\ref{sec:calibration}, assume for the US $160\,000$ ventilators, population $329$ million, ventilation needed by between 5\% and 30\% of the infected (the latter when serious cases become critical due to lack of care). In this case, $K$ is between $\frac{160\,000}{0.05\cdot 329\,000\,000 }=0.0097$ and $\frac{160\,000}{0.3\cdot 329\,000\,000 }=0.0016$. 
The death probability without treatment is between $0.0467$ and $0.3$, but with treatment is $0.023$. The ratio of the death probabilities is assumed to be the loss ratio $\ell_2/\ell_1\in[2,13]$, which defines $\lambda=\frac{\ell_2-\ell_1}{\ell_1}\in[1,12]$. 

Table~\ref{tab:welfare} shows the aggregate losses in the cases in which the fraction needing ventilation is 5\% or 30\% and the ratio of expected individual loss between early and peak infection is 1 or 12. The loss is less (welfare is greater) with early infection in all cases. 
\begin{table}
\caption{Aggregate loss $\int_0^{t_{\max}}I_t[1+\text{\textbf{1}}\set{I_t >K}\lambda (1-\frac{K}{I_t})]dt$ with and without early infection. Parameters $\beta=0.4$, $\gamma=0.1$, average duration of disease $10$ days.} 
\label{tab:welfare}
\centering
\begin{tabular}{ccrr}
$K$ & $\lambda$ & Early infection & No early infection \\
\hline
0.0097 & 1 & 17.83 & 18.59 \\
0.0016 & 1 & 18.38 & 19.09 \\
0.0097 & 12 & 112.1 & 117.5 \\
0.0016 & 12 & 118.7 & 123.5 \\
\hline
\end{tabular}
\end{table}

Under other parameter values, the standard policy may yield greater welfare despite the higher peak of infection. The intuition is that flattening the peak with early infection also spreads it wider, so the infection rate may be above the capacity of the healthcare system for a longer time, depending on the capacity. 
If the capacity is low and the loss rises slowly enough in the amount by which the infected exceed the capacity, then the wider and flatter peak with early infection may be worse. 

A similar welfare-improving early risk increase is possible in a discrete-time SEIR model where people go through an exposed state before transitioning to the infected state. Exposed means non-infectious but carrying the disease, as during an incubation period. Let $E_t$ denote the fraction of the population in the exposed state in period $t$. The duration of the exposed period is $d_E$ and the duration of infection $d_I$. The SEIR model is  
\begin{align*}
&S_t=S_{t-1}-\beta S_{t-1}I_{t-1},\qquad E_t=E_{t-1}+\beta S_{t-1}I_{t-1}-\beta S_{t-1-d_E}I_{t-1-d_E},
\\& I_t=I_{t-1}+\beta S_{t-1-d_E}I_{t-1-d_E}-\beta S_{t-1-d_E-d_I}I_{t-1-d_E-d_I}, 
\\&R_t=R_{t-1}+\beta S_{t-1-d_E-d_I}I_{t-1-d_E-d_I}.
\end{align*} 
For a numerical example of early exposure reducing peak infection, take $\beta=1.8$, $d_E=3$ and $d_I=6$. If $S_0=1-10^{-5}$, $E_0=10^{-5}$, $I_0=R_0=0$, then the peak fraction infected is $I_{\max}=0.7055$, but if $S_0=0.95$, $E_0=0.05$, $I_0=R_0=0$, then $I_{\max}=0.7292$.

\section{Dynamic risk choices}
\label{sec:dynamic}

Time is continuous, indexed by $t$, and the horizon infinite. The discount rate is $\rho\geq0$. 
An individual faces a time-varying risk (a Poisson rate) $r_t$ and suffers a loss of size $\ell_t\in\mathbb{R}_+$ if the risk realizes. The risk may be derived from an SIR model but does not have to be. The loss may be derived from the probability of dying, which depends on the available medical capacity at that time during the epidemic. 

Action $a_t$ changes the risk to $a_t r_t$. 
The actions are: do nothing ($a_t=1$), increase risk 
($a_t=\hat{a}_t>1$) 
or decrease it ($a_t=\check{a}\in(0,1)$), with respective costs $c(1)=0$, $c(\hat{a}_t)=\hat{c}$ and $c(\check{a})=\check{c}$. 
Assume $\hat{a}_t$ is continuous in $t$. 

The individual's value at time $t$ under an optimal strategy is denoted $V_t$. The value is bounded above by $0$ because the cost and the loss both cause negative payoffs. The value is bounded below by $-\max_{\tau\geq t}\ell_t >-\infty$, because doing nothing forever is feasible. The individual's dynamic programming problem is 
\begin{align}
\label{vt}
\rho V_t =\max_{a_t}\set{-c(a_t) +\frac{dV_t}{dt} +a_tr_t(-\ell_t -V_t)}. 
\end{align}
The optimal action is 
\begin{align}
\label{optimalaction}
a_t=\begin{cases}
\hat{a}_t & \text{if } (\hat{a}_t-1)r_t(|V_t|-\ell_t) >\hat{c}, \\
\check{a} & \text{if } (1-\check{a})r_t(\ell_t+V_t) >\check{c}, \\
1 & \text{if } (\hat{a}_t-1)r_t(|V_t|-\ell_t) \leq \hat{c} \text{ and } (1-\check{a})r_t(\ell_t+V_t) \leq \check{c}. \\
\end{cases}
\end{align}
For increasing one's risk to be optimal, a necessary condition is $|V_t|>\ell_t$, thus $\ell_t<\max_{\tau\geq t} \ell_{\tau}$. The interpretation is that the peak of the epidemic has not been reached yet. 
The value is 
\begin{align}
\label{conttvalue}
&V_t= -\int_t^{\infty} \exp\left(-\int_t^x (\rho+a_yr_y)dy \right)[c(a_x)+a_xr_x\ell_x]dx
\\&\label{valuebound}
>\overline{V}_t :=-\int_t^{\infty} \exp\left(-\int_t^x (\rho+\hat{a}_yr_y)dy \right)(\check{a}r_x\ell_x)dx. 
\end{align}
The value is continuous for any bounded $(a_x,r_x,\ell_x)_{x=t}^{\infty}$. 

Assume $r_t,\ell_t$ are continuous in $t$, single-peaked (as is typical for an epidemic) with a peak at $t_{\max}>0$. Assume $\lim_{t\rightarrow\infty}r_t\ell_t =0$ (eventually improved vaccines and treatments make the risk and loss arbitrarily small). The following proposition characterises the individually optimal action paths. 

\begin{prop}
\label{prop:contt}
There exists $t_1\geq0$ s.t.\ $a_t=1\;\forall t\geq t_1$. If $-\overline{V}_t >\frac{\hat{c}}{(\hat{a}_t-1)r_{t}}+\ell_{t}$, then 
there exist $t_0<t_{\downarrow} <t_{\max}<t_1$ s.t.\ the optimal strategy is $a_t=\hat{a}_t$ for $t\in[t,t_0)$, $a_t=1$ for $t\in[t_0,t_{\downarrow}]$, $a_t=\check{a}$ for $t\in(t_{\downarrow},t_1)$ and $a_t=1$ for $t\geq t_1$. 
There exists $k$ depending on $\rho$ and $\frac{\hat{c}}{(\hat{a}_t-1)r_{t}}+\ell_{t}$ s.t.\ if $\max_{x\geq t}\check{a}r_x\ell_x\geq k$, then $-\overline{V}_t >\frac{\hat{c}}{(\hat{a}_t-1)r_{t}}+\ell_{t}$. 
\end{prop}
\begin{proof}
The bound $|V_t|\leq \max_{\tau\geq t}\ell_t$ implies $\lim_{t\rightarrow \infty} |V_t|=0$. Therefore there exists $t_1$ s.t.\ for any $t\geq t_1$, 
$(\hat{a}_t-1)r_t(|V_t|-\ell_t) \leq \hat{c}$ and $(1-\check{a})r_t(\ell_t+V_t) \leq \check{c}$. So the optimal action is $a_t=1$.  

Continuity of $\hat{a}_t,r_t,\ell_t,V_t$ in $t$ implies that the set of times at which a given action is optimal is an interval and the intervals for $\hat{a}_t$ and $\check{a}$ are not adjacent. 

Denote by $t_1\geq 0$ the earliest time after which doing nothing forever is optimal. The value at $t_1$ is 
\begin{align}
\label{vt1}
&V_{t_1} =-\int_{t_1}^{\infty}\exp\left(-\int_{t_1}^{t}(\rho+r_y)dy\right)r_t\ell_tdt.
\end{align} 
If $t_1>0$, then $t_1$ solves 
$(1-\check{a})r_{t_1}(V_{t_1}+\ell_{t_1}) =\check{c}$. 

If $t_1>0$, then there exists $t_{\downarrow0}< t_{\max},t_1$ after which increasing one's risk is never optimal, because $|V_t| <\max_{\tau\geq t}\ell_{\tau}$ and $\ell_{t} <\ell_{t_{\max}}\;\forall t\neq t_{\max}$, so $(\hat{a}_t-1)r_t(|V_t|-\ell_t) <0 <\hat{c}\;\forall t\geq t_{\max}$. 
If $t_{\downarrow0}>0$, then $t_{\downarrow}=\inf t_{\downarrow0}$ solves $(1-\check{a})r_{t_{\downarrow}}(\ell_{t_{\downarrow}}+V_{t_{\downarrow}}) =\check{c}$, which is 
\begin{align*}
&-(1-\check{a})r_{t_{\downarrow}}\int_{t_{\downarrow}}^{t_1} \exp\left(-\int_{t_{\downarrow}}^x (\rho+\check{a}r_y)dy \right)[\check{c}+\check{a}r_x\ell_x]dx 
\\&\notag +(1-\check{a})r_{t_{\downarrow}}\exp\left(-\int_{t_{\downarrow}}^{t_1} (\rho+\check{a}r_y)dy \right)V_{t_1} +(1-\check{a})r_{t_{\downarrow}}\ell_{t_{\downarrow}} =\check{c},
\end{align*} 
because $t_{\downarrow}$ is the earliest time after which the individual first reduces risk and then switches to doing nothing. 
The minimality of $t_{\downarrow}$ and the continuity of $r_t,\ell_t,V_t$ imply there exists $\epsilon>0$ s.t.\ $(1-\check{a})r_{t}(V_{t}+\ell_{t}) \leq\check{c}$ and $(\hat{a}_t-1)r_t(|V_t|-\ell_t) \leq \hat{c}$ for all $t\in(t_{\downarrow}-\epsilon,t_{\downarrow})$. Denote the minimal such $t_{\downarrow}-\epsilon$ by $t_0$. If positive, then it solves $(\hat{a}_{t_0}-1)r_{t_0}(|V_{t_0}|-\ell_{t_0}) = \hat{c}$, equivalently 
\begin{align*}
&\int_{t_0}^{t_{\downarrow}} \exp\left(-\int_{t_0}^x (\rho+r_y)dy \right)r_x\ell_xdx - \exp\left(-\int_{t_0}^{t_{\downarrow}} (\rho+r_y)dy \right)V_{t_{\downarrow}} =\frac{\hat{c}}{(\hat{a}_{t_0}-1)r_{t_{0}}} +\ell_{t_{0}}. 
\end{align*}
If $t_0>0$, then by the continuity of $r_t,\ell_t,V_t$, there exists $\eta>0$ s.t.\ $(\hat{a}_t-1)r_t(|V_t|-\ell_t) > \hat{c}$ for all $t\in(t_{0}-\eta,t_{0})$. 
Based on~(\ref{valuebound}), $V_t >\overline{V}_t$. If $-\overline{V}_t >\frac{\hat{c}}{(\hat{a}_t-1)r_{t}}+\ell_{t}$, then $a_t=\hat{a}_t$ is optimal. Sufficient is that $\max_{x\geq t}\check{a}r_x\ell_x$ is above a cutoff which increases in $\rho$ and $\frac{\hat{c}}{(\hat{a}_t-1)r_{t}}+\ell_{t}$. 
\end{proof}

Intuitively, if the risk and loss eventually vanish, then doing nothing becomes optimal from some time onward. If the peak of the epidemic is severe enough, then reducing risk is optimal in a time interval around the peak and increasing one's risk is optimal early on when the risk and loss are low compared to the peak. Between any intervals of raising and reducing risk, there is an interval of doing nothing because the incentives are continuous in time.

Next, the continuous time individual optimization problem is calibrated to the epidemiologic parameters of the 2019 coronavirus. Increasing one's risk early on in the epidemic turns out to be optimal, both when other individuals do nothing and when all individuals respond optimally (so the risk and loss are derived from the decisions in equilibrium).  


\subsection{Calibration of the dynamic individual decision}
\label{sec:dynamiccalibration}

The Mathematica code for the simulations in this section is available on the author's website \url{https://sanderheinsalu.com/}. 
First, an individual decision is calibrated assuming that the rest of society does not respond to the incentives to raise or reduce risk. After that, the background risk is derived in equilibrium from individual choices. 

The unit of time is one day. The discount factor is assumed zero because the time horizon of the pandemic is less than one year. 
The risk function $r_t$ for $t\leq 10^4$ is the rate $\beta I_t$ at which susceptible people get infected in the SIR model of \cite{kermack+mckendrick1927}.\footnote{Fraction $S_t$ is susceptible, $I_t$ infected and $R_t=1-S_t-I_t$ recovered, resistant or dead at time $t$. The fractions evolve as $\frac{dS_t}{dt} =-\beta S_tI_t$, $\frac{dI_t}{dt} =\beta S_tI_t-\gamma I_t$ and $\frac{dR_t}{dt} =\gamma I_t$. Thus a susceptible gets infected at rate $-\frac{dS_t}{S_tdt}=\frac{\beta S_tI_t}{S_t}$. 
}
The parameters are $I_0=10^{-6}$, $\beta=0.4$ and $\gamma=0.1$ (initially, one in a million people is infected and the average duration of the illness is $10$ days). At the peak, 40\% of the population is infected. 
After $t=10^4$, the risk and loss are assumed zero. The interpretation is vaccine and treatment development over the $10^4$ days (27 years). 

The per-person loss function $\ell_t$ for $t\leq 10^4$ is $1+\text{\textbf{1}}\set{I>K}\lambda (1-\frac{K}{I})$, so the loss rises if the infected fraction $I$ of the population exceeds the capacity $K$ of the medical system. 
As in Section~\ref{sec:calibration}, assume for the US $160\,000$ ventilators, population $329$ million, ventilation needed by between 5\% and 30\% of the infected (the latter when serious cases become critical due to lack of care). In this case, $K$ is between $\frac{160\,000}{0.05\cdot 329\,000\,000 } =0.0097$ and $\frac{160\,000}{0.3\cdot 329\,000\,000 } =0.0016$. 
The death probability without treatment is between $0.0467$ and $0.3$, but with treatment is $0.023$. The ratio of the death probabilities is assumed to be the loss ratio $\ell_2/\ell_1\in[2,13]$, which defines $\lambda=\frac{\ell_2-\ell_1}{\ell_1}\in[1,12]$. 

The cost of reducing or increasing risk is assumed $\check{c} =\hat{c} =0.01$. The unit in which costs and losses are measured is the death probability conditional on infection when the medical system functions normally. 
The factor by which risk can be decreased or increased is assumed to be $2$, so that $\check{a}=0.5$ and $\hat{a}=2$. 

At $\lambda=1$ and $K=0.0097$, an individual optimally does nothing after $t_1=84$, reduces risk between $t_{\downarrow}=44$ and $t_1$, does nothing between $t_0=31$ and $t_{\downarrow}$ and increases risk when $t\in[0,t_0)$. 
If $\lambda=12$ and $K=0.0097$, then $t_1=94$, $t_{\downarrow}=38$ and $t_0=26$. 
If $\lambda=1$ and $K=0.0016$, then $t_1=85$, $t_{\downarrow}=44$ and $t_0=25$.
If $\lambda=12$ and $K=0.0016$, then $t_1=103$, $t_{\downarrow}=36$ and $t_0=25$. In all cases, increasing one's risk initially is optimal. Reducing the cost $\hat{c}$ or increasing the effectiveness $\hat{a}_t$ of raising the risk widens the time interval during which raising one's risk is the best response. On the other hand, reducing $\check{c}$ or $\check{a}$ shortens this interval. 

An equilibrium of the dynamic model consists of $S_t,I_t,R_t,r_t,V_t$ and $a_t^*$ such that given $a_t^*$, the aggregate variables evolve according to $\frac{dS_t}{dt} =-a_t^*\beta S_tI_t$, $\frac{dI_t}{dt} =a_t^*\beta S_tI_t-\gamma I_t$, $\frac{dR_t}{dt} =\gamma I_t$, the risk is $r_t=a_t^*\beta I_t$ and the value function is~(\ref{conttvalue}), 
and given $r_t,V_t$, the individual's optimal policy is~(\ref{optimalaction}). 

Equilibrium can be found numerically. In it, the individual's optimal policy remains qualitatively similar to the one calibrated above. For example, if $\lambda=1$ and $K=0.0097$, then each individual optimally does nothing after $t_1=64$, reduces risk between $t_{\downarrow}=20$ and $t_1$, does nothing between $t_0=18$ and $t_{\downarrow}$ and increases risk when $t\in[0,t_0)$. 
The increased risk taken by all individuals early in the epidemic compresses the epidemic in time. The cutoff times $t_0,t_{\downarrow},t_{1}$ and the peak time $t_{\max}$ occur earlier tan when all individuals do nothing. The ex ante value if everyone does nothing is $-1.882$, but the equilibrium value $-1.98$. Welfare is lower in equilibrium than under doing nothing because the transmission externality outweighs the peak load externality.

\section{Conclusion}
\label{sec:conclusion} 

Imperfect measures to reduce infection, even if costless, may be suboptimal. By contrast, deliberately infecting a fraction of the population early, even at extra cost, increases welfare under some parameter values that match the Covid-19 data. 

The policy of infecting part of the population early is difficult to implement. Even relatively safe vaccines are refused. A live unattenuated vaccine (the wild-type disease agent) likely generates even more opposition, although it is probably more immunogenic than safer vaccines. 

Another practical problem is evaluating the optimality of any infection control policy, including early infection, at the start of an epidemic. Information is limited, so the parameters of the disease that determine the optimal control method are imprecisely estimated. Given the risk of worsening the epidemic by infecting too large a fraction of the population early, the principle of `first, do no harm' argues against the early infection strategy. 

An argument against the early infection policy is that the epidemic may not be severe. If its peak still leaves spare capacity in the healthcare system, then there is no benefit to flattening and spreading the peak. 
The benefit relies on (at least partial) immunity generated by early infection, so is negative for diseases that lie dormant in the organism (herpes, HIV, tuberculosis, cancer) and flare up later. For diseases with a long duration (HIV), discounting over the course of the epidemic is a significant factor in preferences and welfare. Discounting reduces the benefit of early infection, which moves the peak of the epidemic earlier. 

Risk aversion increases the benefit of the early infection policy for at least two reasons. A certain early infection with the same expected loss is preferable to an uncertain later one. The probability of death is closer to $\frac{1}{2}$ at the peak of the epidemic than early on, thus  generates a lottery with higher variance, which a risk averse decision maker dislikes. 
Prospect theory preferences, on the other hand, are risk loving for lotteries with negative outcomes, thus reduce the benefit of early infection. 

The benefit calculation for the SIR model in this paper assumes random mixing. 
If instead the initially deliberately infected volunteers are quarantined until they recover, then the welfare gain relative to standard policies is greater.

\bibliographystyle{ecta}
\bibliography{teooriaPaberid} 
\end{document}